\documentclass[11pt,a4paper]{article}
\usepackage{calligra}

\usepackage[T1]{fontenc}

\usepackage{placeins}
\usepackage{tikz}
\usepackage{tikz-3dplot}
\usepackage{pgfplots}
\usepackage{tikzscale}
\usepackage{subcaption}
\usepackage[utf8]{inputenc}
\usepackage{amsmath}
\usepackage{amssymb}
\usepackage{graphicx}
\usepackage{epstopdf}
\usepackage[amsmath,thref,thmmarks,hyperref]{ntheorem}
\usepackage{dsfont}
\usepackage{fullpage}
\usepackage{multirow}
\usepackage[affil-it]{authblk}
\usepackage{braket}

\newcommand{\ketbra}[2]{\ket{#1}\!\bra{#2}}
\usepackage{enumerate}
\usepackage[shortlabels]{enumitem}

\usepackage{mathtools}
\usepackage{nicefrac}

\usepackage[backend=biber,firstinits=true,doi=false,url=false,maxbibnames=5,style=alphabetic]{biblatex}
\renewbibmacro{in:}{}
\usepackage[toc,page]{appendix}

\usepackage[bookmarks]{hyperref}

\theoremseparator{.}
\theorempreskip{0.7\topsep}
\newtheorem{theorem}{Theorem}[section]
\newtheorem{lemma}[theorem]{Lemma}
\newtheorem{proposition}[theorem]{Proposition}
\newtheorem{corollary}[theorem]{Corollary}
\theorembodyfont{\normalfont}
\theoremstyle{plain}
\theoremsymbol{\ensuremath{\diamond}}
\newtheorem{definition}[theorem]{Definition}
\newtheorem*{remark}{Remark}
\makeatletter
\newtheoremstyle{MyNonumberplain}%
  {\item[\theorem@headerfont\hskip\labelsep ##1\theorem@separator]}%
  {\item[\theorem@headerfont\hskip\labelsep ##3\theorem@separator]}
\makeatother
\theoremstyle{MyNonumberplain}
\theoremheaderfont{\itshape}
\theorembodyfont{\upshape}
\theoremsymbol{\ensuremath{\Box}}
\newtheorem{proof}{Proof}

\usepackage{cleveref}

\usepackage{accents}

\newcommand{\cH}{\mathcal{H}}
\newcommand{\cD}{\mathcal{D}}
\newcommand{\D}{\mathcal{D}}
\newcommand{\Dpure}{\D_\text{pure}}

\DeclareMathAlphabet{\mathpzc}{OT1}{pzc}{m}{it}
 \newcommand{\cs}{\mathpzc{s}}

\newcommand{\cB}{\mathcal{B}}
\newcommand{\Bsa}{\mathcal{B}_\text{sa}}

\newcommand{\one}{\mathds{1}}
\newcommand{\eps}{\varepsilon}

\newcommand{\RR}{\mathbb{R}}
\newcommand{\R}{\mathbb{R}}

\DeclareMathOperator{\tr}{Tr}

\DeclareMathOperator{\id}{id}

\DeclareMathOperator{\argmax}{argmax}

\DeclareMathOperator{\diag}{diag}

\renewcommand{\epsilon}{\varepsilon}
\newcommand{\Be}{B_\varepsilon}

\renewcommand{\one}{\mathds{1}}
\renewcommand{\H}{\mathcal{H}}
\DeclareMathOperator{\spec}{spec}

\newcommand{\mmm}{\mathcal{M}} % majorization-minimizer map

\begin{document}
\title{Tight uniform continuity bound for a family of entropies}

\author{Eric P. Hanson\thanks{Email: \texttt{ephanson@damtp.cam.ac.uk}} }
\author{Nilanjana Datta\thanks{Email: \texttt{n.datta@damtp.cam.ac.uk}}}
\affil{\small Department of Applied Mathematics and Theoretical Physics, Centre for Mathematical Sciences\\University of Cambridge, Cambridge~CB3~0WA, UK}

\maketitle

\begin{abstract}
We prove a tight uniform continuity bound for a family of entropies which includes the von Neumann entropy, the Tsallis entropy and the $\alpha$-R\'enyi entropy, $S_\alpha$, for $\alpha\in (0,1)$. 
We establish necessary and sufficient conditions for equality in the continuity bound and prove that these conditions are the same for every member of the family. Our result builds on 
recent work in which we constructed a state which was majorized by every state in a neighbourhood ($\eps$-ball) of a given state, and thus was the minimal state in majorization order in the $\eps$-ball. 
This minimal state satisfies a particular semigroup property, which we exploit to prove our bound.
\end{abstract}

\section{Introduction}

Entropies play a fundamental role in quantum information theory as characterizations of the optimal rates of information theoretic tasks, and as measures of uncertainty. The mathematical properties of entropic functions therefore have important physical implications. The von Neumann entropy $S$, for instance, as a function of $d$-dimensional quantum states, is strictly concave, continuous, and is bounded by $\log d$. As the von Neumann entropy characterizes the optimal rate of data compression for a memoryless quantum information source \cite{Schumacher}, continuity of the von Neumann entropy, for example, implies that the quantum data compression limit is continuous in the source state. The $\alpha$-R\'enyi entropies $S_\alpha$ are parametrized by $\alpha \in (0,1)\cup (1,\infty)$, and are a generalization of the von Neumann entropy in the sense that $\lim_{\alpha\to 1} S_\alpha = S$. The $\alpha$-R\'enyi entropy has been used to bound the quantum communication complexity of distributed information-theoretic tasks \cite{vDH02}, can be interpreted in terms of the free energy of a quantum or classical system \cite{Baez11}, and is the fundamental quantity defining the entanglement $\alpha$-R\'enyi entropy \cite{WMVF16}.

In fact, the $\alpha$-R\'enyi entropies are members of a large family of entropies called the $(h,\phi)$-entropies, which are parametrized by two functions $h,\phi$ on $\R$ subject to certain constraints (see \Cref{sec:notation_and_def}). This family includes the Tsallis entropies \cite{Tsallis1988} and the unified entropies (considered by Rastegin in \cite{Ras2011}). Note that the $(h,\phi)$-entropy of a quantum state is the classical $(h,\phi)$-entropy of its eigenvalues, and therefore the results here apply equally well to probability distributions on finite sets.

Continuity is a useful property of entropic functions, particularly when cast in the form of a \emph{uniform continuity bound}: given two $d$-dimensional states which are at a trace distance of at most $\eps \in (0,1)$, this provides a bound on their entropy difference entirely in terms of $\eps$ and $d$. Fannes first proved a uniform continuity bound for the von Neumann entropy \cite{Fannes1973}. This bound was improved to a tight form by Audenaert \cite{Audenaert07} and is often called the called the Audenaert-Fannes bound (see also \cite[Theorem 3.8]{PetzQITbook}). Rastegin proved similar continuity bounds for the unified entropies, which include the $\alpha$-R\'enyi entropies and Tsallis entropies, but the resulting bounds are not known to be tight \cite{Ras2011}. Recently, Chen et al proved continuity bounds for the $\alpha$-R\'enyi entropy for $\alpha \in (0,1)\cup(1,\infty)$ using techniques similar to Audenaert's proof of the Audenaert-Fannes bound \cite{Renyi-CMNF}, but the resulting bounds are known to be not tight \cite{comm-CMNF}.

In \cite{HD17}, we considered {\em{local continuity bounds}}. Given a $d$-dimensional quantum state $\sigma$, a \emph{local} continuity bound of an entropic function $H$ at $\sigma$ is a bound on the entropy difference $|H(\omega) - H(\sigma)|$ for any $\omega$ in an $\eps$-ball around $\sigma$, which depends not only on $\eps$ and $d$ but also on the state $\sigma$ itself. These local bounds hence incorporate additional information about the state $\sigma$, for example, its spectrum, to yield a bound which is tighter than a uniform continuity bound. By finding maximizers and minimizers of the \emph{majorization order} on $d$-dimensional quantum states over the $\eps$-ball around $\sigma$, local bounds were obtained for any $(h,\phi)$-entropy, in fact, for any Schur concave entropic function in \cite{HD17}. 

Given a quantum state $\sigma$ and $\eps\in(0,1]$, we denote the $\eps$-ball in trace distance around $\sigma$ by $\Be(\sigma)$ (defined by \cref{eq:eps-ball} below). For a given $\sigma$ and $\eps$, there exist two quantum states $\sigma^*_\sigma,\sigma_{*,\eps} \in \Be(\sigma)$ such that for any $\omega\in \Be(\sigma)$ centered at $\sigma$,
\[
\sigma^*_\eps  \prec \omega \prec \sigma_{*,\eps}
\]
where $\prec$ denotes the majorization order (defined in \Cref{sec:notation_and_def}). In \cite{HD17}, this fact was proved by explicit construction of these states, using the notation $\rho^*_\eps(\sigma)$ for $\sigma^*_\eps$
and $\rho_{*,\eps}(\sigma)$ for $\sigma_{*,\eps}$. These states were also independently found by Horodecki, Oppenheim, and Sparaciari \cite{HO17approxmaj}, and considered in the context of thermal majorization \cite{Remco, MNW17_thermal}. In \cite{HD17} we also established that the minimal state $\rho^*_\eps(\sigma)\equiv \sigma^*$ in the majorization order, satisfied a semigroup property:  $\rho^*_{\eps_1+\eps_2}(\sigma) = \rho^*_{\eps_1}( \rho^*_{\eps_2}(\sigma))$. This property plays a key role in the proof of the main results of this paper.

 In \Cref{sec:notation_and_def} we introduce the basic notation and definitions and in \Cref{sec:main_results} we state our main results. The proof strategy is described in \Cref{sec:proof-strat}  and in \Cref{sec:Lambda-eps} the construction of the minimal state (in the majorization order), $\sigma_\eps^*$, which we use in our proof, is formulated. \Cref{sec:proof-Delta-eps-Schur-convex} consists of a proof of the main technical result \Cref{thm:Delta_eps_Schur_convex} and employs certain lemmas which are proved in \Cref{sec:proof_lemmas}. In \Cref{sec:elem-prop-concave-fun}, we recall an elementary property of concave functions.

\section{Notation and definitions \label{sec:notation_and_def}}
Let $\cH$ denote a finite-dimensional Hilbert space, with $\dim \H = d$, $\cB(\cH)$ the set of (bounded) linear operators on $\cH$, and  $\Bsa(\cH)$ the set of self-adjoint linear operators on $\cH$. A quantum state (or density matrix) is a positive semidefinite element of $\cB(\cH)$ with trace one. Let $\D(\cH)$ be the set quantum states on $\cH$. We denote the completely mixed state by $\tau := \frac{\one}{d}$. A pure state is a rank-1 density matrix; we denote the set of pure states by $\Dpure(\cH)$. For two quantum states $\rho,\sigma \in \D(\cH)$, the \emph{trace distance} between them 
is given by
\[
T(\rho,\sigma) = \frac{1}{2}\|\rho-\sigma\|_1.
\]
We define the $\eps$-ball around $\sigma\in \D(\cH)$ as the set
\begin{equation}
\Be(\sigma) = \{ \omega\in \D(\cH): T(\omega,\sigma) \leq \eps \}. \label{eq:eps-ball}
\end{equation}
For any $A\in \Bsa(\cH)$, let $\lambda_+ (A)$ and $\lambda_-(A)$ denote the maximum and minimum eigenvalue of $A$, respectively, and $k_+(A)$ and $k_-(A)$ denote their multiplicities. Let $\lambda_j(A)$ denote the $j$th largest eigenvalue, counting multiplicity; that is, the $j$th element of the ordering
\[
\lambda_1(A)\geq \lambda_2(A) \geq \dotsm \geq \lambda_d(A).
\]
We set $\vec \lambda (A) := (\lambda_i(A))_{i=1}^d \in\R^d$ and denote the set of eigenvalues of $A\in \Bsa(\cH)$ by $\spec A \subset \R$.

Given $x\in \R^d$, write $x^\downarrow = (x^\downarrow_j)_{j=1}^d$ for the permutation of $x$ such that $x^\downarrow_1 \geq x^\downarrow_2 \geq \dotsm \geq x^\downarrow_d$. For $x,y\in \R^d$, we say $x$ \emph{majorizes} $y$, written $x \succ y$, if 
	\begin{equation} \label{def:majorize}
	 \sum_{j=1}^k x^\downarrow_j \geq \sum_{j=1}^k y^\downarrow_j \quad \forall k=1,\dotsc,d-1, \quad \text{and}\quad \sum_{j=1}^d x^\downarrow_j = \sum_{j=1}^d y^\downarrow_j.
	 \end{equation} 
Given two states $\rho,\sigma\in \cD$, we say $\sigma$ majorizes $\rho$, written $\rho\prec \sigma$ if $\vec\lambda(\rho) \prec \vec\lambda(\sigma)$. We say that $\varphi: \D \to \R$ is \emph{Schur convex} if $\varphi(\rho)\leq \varphi(\sigma)$ for any $\rho,\sigma\in \cD$ with $\rho \prec \sigma$. If $\varphi(\rho) < \varphi(\sigma)$ for any $\rho,\sigma\in \cD$ such that  $\rho \prec \sigma$, and $\rho$ is not unitarily equivalent to $\sigma$, then $\varphi$ is \emph{strictly Schur convex}. We say $\varphi$ is Schur concave (resp.~strictly Schur concave) if $(-\varphi)$ is Schur convex (resp.~strictly Schur convex).

Let $h: \R\to \R$ and $\phi: [0,1] \to \R$ with $\phi(0) = 0$ and $h(\phi(1)) = 0$, such that either $h$ is strictly increasing and $\phi$ strictly concave, or $h$ strictly decreasing and $\phi$ strictly convex. Then the $(h,\phi)$-entropy, $H_{(h,\phi)}$, is defined by
\begin{equation} \label{eq:def_h-phi-entropy}
H_{(h,\phi)}(\rho) := h(\tr[\phi(\rho)])
\end{equation}
where $\phi$ is defined on $\D(\cH)$ by functional calculus, i.e.~given the eigen-decomposition $\rho = \sum_i \lambda_i(\rho) \pi_i$, we have $\phi(\rho) = \sum_i \phi(\lambda_i(\rho)) \pi_i$. Every $(h,\phi)$-entropy is strictly Schur concave and unitarily invariant; moreover, if $h$ is concave, then $H_{(h,\phi)}$ is concave \cite{Bosyk2016}. Here, we are most interested in the following three examples of $(h,\phi)$ entropies:
\begin{itemize}
	\item The von Neumann entropy
	\[
	 S(\rho) = -\tr (\rho \log \rho).
	 \] 
	 $S$ is the $(h,\phi)$ entropy with $h=\id$, i.e., $h(x) = x$ for $x\in \RR$, and with $\phi(x) = - x \log x$ for $x\in [0,1]$. The von Neumann entropy satisfies the following tight continuity bound known as the Audenaert-Fannes bound \cite{Audenaert07} (see also \cite[Theorem 3.8]{PetzQITbook}). Given $\eps\in(0,1]$ and $\rho,\sigma\in \cD(\cH)$ with $T(\rho,\sigma)\leq \eps$, 
	 \begin{equation}
 |S(\rho) - S(\sigma) | \leq \begin{cases}
 \epsilon \log (d-1) + h(\epsilon) & \text{if } \epsilon < 1 - \tfrac{1}{d} \\
	 \log d & \text{if } \epsilon \geq 1 - \tfrac{1}{d}
 \end{cases} \label{eq:Audenaert-Fannes_bound}
 \end{equation}
where $h(\eps) := - \eps \log \eps - (1-\eps) \log (1-\eps)$ denotes the binary entropy.
\item The $q$-Tsallis entropy for $q\in(0,1)\cup(1,\infty)$,
\[
T_q(\rho) = \frac{1}{1-q}[\tr(\rho^q)-1].
\]
$T_q$ can be written as the $(h,\phi)$-entropy with $h(x) = x - \frac{1}{1-q}$ for $x\in \R$ and $\phi(x) = \frac{1}{1-q}x^q$. With these choices, $h$ is strictly increasing and affine (and therefore concave) and $\phi$ is strictly concave, for all $q\in (0,1)\cup(1,\infty)$.
	\item The $\alpha$-R\'enyi entropy for $\alpha \in (0,1)\cup(1,\infty)$,
	\[
	S_\alpha(\rho) = \frac{1}{1-\alpha} \log \left(\tr \rho^\alpha\right).
	\]
	$S_\alpha$ is the $(h,\phi)$-entropy with $h(x) = \frac{1}{1-\alpha}\log x$ for $x\in \RR$ and $\phi(x) = x^\alpha$ for $x\in [0,1]$. For $\alpha\in (0,1)$, $h$ is concave and strictly increasing and $\phi$ is strictly concave. For $\alpha >1$, $h$ is convex and strictly decreasing, and $\phi$ is strictly convex. It is known that $\lim_{\alpha \to 1}S_\alpha(\rho) = S(\rho)$.
\end{itemize}
In the above, all logarithms are taken to base $2$.

\section{Main results \label{sec:main_results}}
\begin{theorem}[Uniform continuity bounds] Let $H_{(h,\phi)}$ be an $(h,\phi)$-entropy, defined through \eqref{eq:def_h-phi-entropy}) with $h$ concave and $\phi$ strictly concave.
 For $\eps \in (0,1]$ and any states $\rho,\sigma\in \D(\cH)$  such that $\frac{1}{2}\|\rho-\sigma\|_1\leq \eps$, we have
 \begin{equation}
 | H_{(h,\phi)}(\rho) - H_{(h,\phi)}(\sigma) | \leq \begin{cases}
h( \phi(1-\eps) + (d-1) \phi( \frac{\eps}{d-1})) & \eps < 1-\frac{1}{d}\\
h(d\phi(\frac{1}{d})) & \eps \geq 1 - \frac{1}{d}
\end{cases} \label{eq:hphi-uniform-bound} 
 \end{equation}
 and in particular, for $\alpha \in (0,1)$,
 \begin{equation}
 | S_\alpha(\rho) - S_\alpha(\sigma) | \leq \begin{cases}
\frac{1}{1-\alpha} \log ( (1-\eps)^\alpha + (d-1)^{1-\alpha} \eps^\alpha) & \eps < 1-\frac{1}{d}\\
\log d& \eps \geq 1 - \frac{1}{d}
\end{cases} \label{eq:uniform_Renyi_bound}
 \end{equation}
 and for $q \in (0,1)\cup(1,\infty)$,
 \begin{equation}
 | T_q(\rho) - T_q(\sigma) | \leq \begin{cases}
\frac{1}{1-q}  ( (1-\eps)^q + (d-1)^{1-q} \eps^q - 1) & \eps < 1-\frac{1}{d}\\
\frac{ d^{1-q}-1}{1-q}& \eps \geq 1 - \frac{1}{d},
\end{cases} \label{eq:uniform_Tsallis_bound} 
 \end{equation}
where $d =\dim \cH$. Moreover, equality in \eqref{eq:hphi-uniform-bound}, \eqref{eq:uniform_Renyi_bound}, or \eqref{eq:uniform_Tsallis_bound} occurs if and only if one of the two states (say, $\sigma$) is pure, and either
\begin{enumerate}
	\item $\eps < 1 - \frac{1}{d}$ and $\vec \lambda(\rho) = ( 1- \eps, \frac{\eps}{d-1},\dotsc, \dotsc, \frac{\eps}{d-1})$, or
	\item  $\eps \geq 1- \frac{1}{d}$, and $\rho = \tau := \frac{\one}{d}$.
\end{enumerate}\label{thm:hphi-GCB}
\end{theorem}

\begin{remark}
~\begin{itemize}
	\item When \eqref{eq:hphi-uniform-bound} is applied to the von Neumann entropy $S$, one recovers the Audenaert-Fannes bound, \eqref{eq:Audenaert-Fannes_bound}, with equality conditions. The sufficiency of these equality conditions were shown in \cite{Audenaert07}, and their necessity was recently derived in \cite{HD17} by an analysis of the proof of the bound presented in \cite[Thm. 3.8]{PetzQITbook} and \cite{Winter16}, which involves a coupling argument. We establish that these necessary and sufficient conditions are the same for every $(h,\phi)$-entropy satisfying the conditions of the theorem.
	\item The inequality \eqref{eq:uniform_Renyi_bound}  reduces to the Audenaert-Fannes bound \eqref{eq:Audenaert-Fannes_bound}  when the limit $\alpha \to 1$ is taken on both sides of it.
	\item The bound \eqref{eq:uniform_Tsallis_bound} appeared in \cite{Renyi-CMNF} as Lemma 1.2, and was derived with a different method. However, the equality conditions were not established.
	\item See \Cref{fig:bound_comparison} for a comparison of our uniform continuity bound for the $\alpha$-R\'enyi entropy, \eqref{eq:uniform_Renyi_bound}, for $\alpha = \frac{1}{2}$, with those obtained in \cite{Ras2011} and \cite{Renyi-CMNF}.
\end{itemize}

\end{remark}

\begin{figure}[ht]
\centering
\includegraphics[width = .75\textwidth]{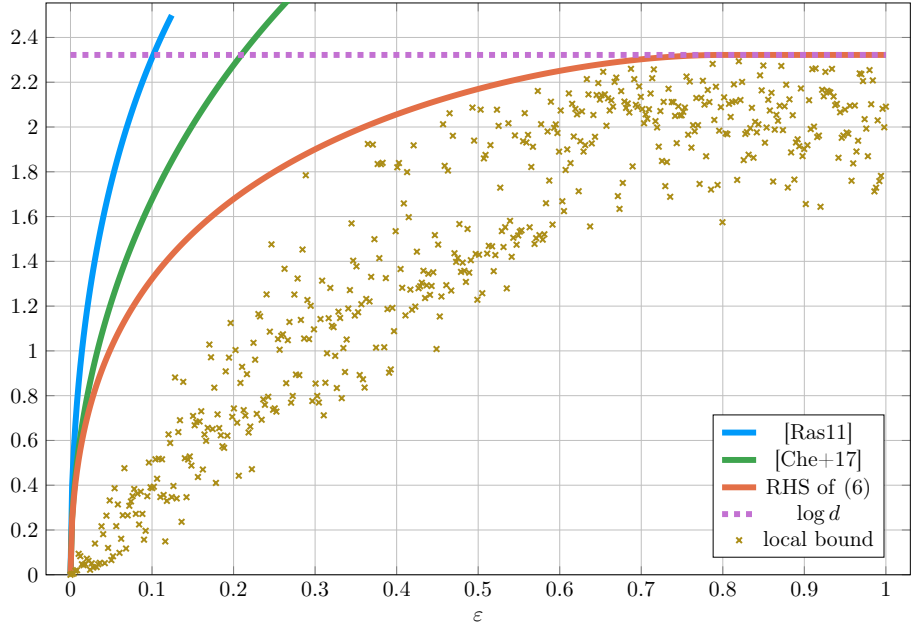}
\caption{In dimensions $d=5$, for $\alpha = \frac{1}{2}$, the bound given by the right-hand side of \eqref{eq:uniform_Renyi_bound} is compared to the bounds given by Equation (7) of \cite{Renyi-CMNF} and by Equation (27) of \cite{Ras2011}. We also include the trivial bound $\log d$, as well as 500 points corresponding the local bounds found in \cite{HD17} computed at (uniformly) randomly chosen $\sigma\in \D(\cH)$.} \label{fig:bound_comparison}
\end{figure}

\section{Proof strategy} \label{sec:proof-strat}

Given a state $\sigma\in \D(\cH)$ and $\eps\in (0,1]$, one can construct two states $\sigma_\eps^*,\sigma_{*,\eps}\in \Be(\sigma)$ such that
\begin{equation} \label{eq:maj-order}
\sigma_\eps^* \prec \omega \prec \sigma_{*,\eps}
\end{equation}
for any $\omega\in \Be(\sigma)$. This was done in \cite{HD17}, with the notation $\rho_\eps^*(\sigma)$ (resp.~$\rho_{*,\eps}(\sigma)$) to denote $\sigma_\eps^*$ (resp.~$\sigma_{*,\eps}$). These states were also independently found in \cite{HO17approxmaj}, and considered in the context of thermal majorization in \cite{Remco,MNW17_thermal}. The proof of our main result relies on the form of $\sigma_\eps^*$ and its properties. An explicit construction of $\sigma_\eps^*$ is given in \Cref{sec:Lambda-eps}, and its properties are described in \Cref{prop:properties_of_Lambda_eps}.

Consider an $(h,\phi)$ entropy $H_{(h,\phi)}$, and let $\epsilon\in (0,1]$, and $\rho,\sigma\in \cD(\cH)$ with $T(\rho,\sigma)\leq \epsilon$. If $H_{(h,\phi)}(\rho) \geq H_{(h,\phi)}(\sigma)$, then since $\rho\in \Be(\sigma)$,
\begin{equation}
|H(\rho) - H(\sigma)| = H(\rho) - H(\sigma) \leq \max_{\omega\in \Be(\sigma)} H(\omega) - H(\sigma) = H(\sigma_\eps^*) - H(\sigma) \label{eq:Hrho-Hsig}
\end{equation}
where the last equality follows from the first majorization relation in \cref{eq:maj-order} and the strict Schur concavity of $H_{(h,\phi)}$. Similarly, if $H_{(h,\phi)}(\sigma) \geq H_{(h,\phi)}(\rho)$, \cref{eq:Hrho-Hsig} holds with $\sigma$ (resp.~$\sigma_\eps^*$) replaced by $\rho$ (resp.~$\rho_\eps^*$). Hence, in general,
\begin{equation} \label{eq:UBHrho-Hsig_by_max}
|H_{(h,\phi)}(\rho) - H_{(h,\phi)}(\sigma)| \leq \max \{\Delta_\eps(\rho),\Delta_\eps(\sigma)\} \leq \max_{\omega \in \cD(\cH)} \Delta_\eps(\omega),
\end{equation}
where 
\begin{equation} \label{eq:def_Delta-eps}
\begin{aligned}
\Delta_\eps : \qquad \cD(\cH) &\to \R_{\geq 0}\\
\omega &\mapsto H_{(h,\phi)}\circ \mmm_\eps(\omega) - H_{(h,\phi)}(\omega),
\end{aligned}
\end{equation}
and $\mmm_\eps$ is the \emph{majorization-minimizer map},
\begin{equation} \label{eq:def_mmm-eps}
\begin{aligned}
\mmm_\eps : \qquad \cD(\cH) &\to  \cD(\cH)\\
\omega &\mapsto \omega_\eps^*.
\end{aligned}
\end{equation}
This map is defined explicitly by \cref{def:Lambdaeps} in \Cref{sec:Lambda-eps}.
 Note that $\Delta_\eps(\omega)\geq 0$ for $\omega\in \cD(\cH)$ follows from the Schur concavity of the $(h,\phi)$-entropy. To prove \Cref{thm:hphi-GCB}, it remains to maximize $\Delta_\eps$ over $\cD(\cH)$. 

We show that for $(h,\phi)$-entropies for which $h$ is concave and $\phi$ (strictly) convex, $\Delta_\eps$ is a Schur convex function on $\cD(\cH)$, which is our main technical result. We u defer its proof to \Cref{sec:proof-Delta-eps-Schur-convex}.
\begin{theorem} \label{thm:Delta_eps_Schur_convex}
Assume $h$ is concave and $\phi$ is strictly concave. Let $\eps\in (0,1]$. Then $\Delta_\eps: \D(\cH)\to  \R_{\geq 0}$ is  Schur convex. That is, if $\rho\prec \sigma$,
\[
\Delta_\eps(\rho) \leq \Delta_\eps(\sigma).
\]
Moreover, $\Delta_\eps(\rho) = \Delta_\eps(\sigma)
$ implies $\lambda_+(\rho) = \lambda_+(\sigma)$. Lastly, if $h$ is strictly concave, then $\Delta_\eps$ is strictly Schur convex.
\end{theorem}
Note that if $h$ is not strictly concave, $\Delta_\eps$ need not be strictly Schur convex. In fact, for the von Neumann entropy we can find a counterexample to strict Schur convexity of $\Delta_\eps$. Setting $\rho = \diag(0.1,0.2,0.2,0.5)$ and $\sigma = \diag(0.1, 0.15, 0.25, 0.5)$ yields $\rho \prec \sigma$ and that $\rho$ and $\sigma$ are not unitarily equivalent. However, for $\eps \leq 0.05$, we have  $\Delta_\eps(\rho) = \Delta_\eps(\sigma)$. 
\begin{corollary}If $h$ is concave, $\phi$ strictly concave, and $\eps\in (0,1]$, then $\Delta_\eps$ achieves a maximum on $\D(\cH)$, and moreover $\argmax  \Delta_\eps = \Dpure(\cH) $. \label{cor:Delta-eps-max-at-pure}
\end{corollary}
\begin{proof}	
Since any pure state $\psi$ satisfies $\rho\prec \psi$ for every $\rho\in \D(\cH)$, we have $\Delta_\eps(\psi) \geq \Delta_\eps(\rho)$ for every $\rho\in \D(\cH)$. Therefore,  $\Dpure(\cH) \subset \argmax \Delta_\eps$. On the other hand, if $\omega\in \D(\cH)$ has $\omega \in \argmax \Delta_\eps$, then
\[
\Delta_\eps(\omega) = \Delta_\eps(\psi)
\]
for a pure state $\psi$. Therefore, $\lambda_+(\omega) = \lambda_+(\psi)=1$, and $\omega$ must be a pure state.
\end{proof}

\begin{figure}[ht]
\centering
\includegraphics[width=.408\textwidth]{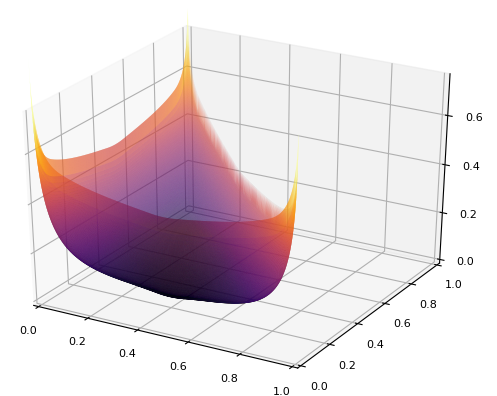}
\includegraphics[width=.492\textwidth]{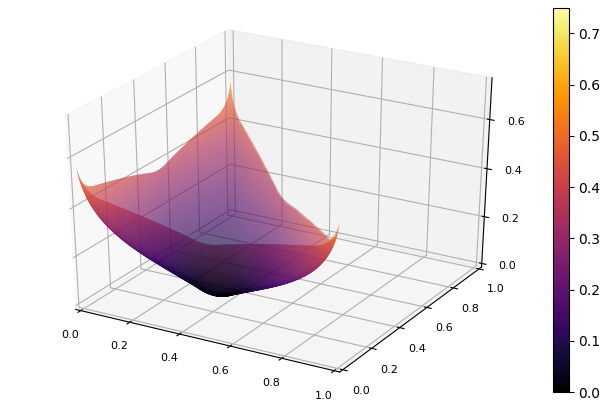}
\includegraphics[width=.408\textwidth]{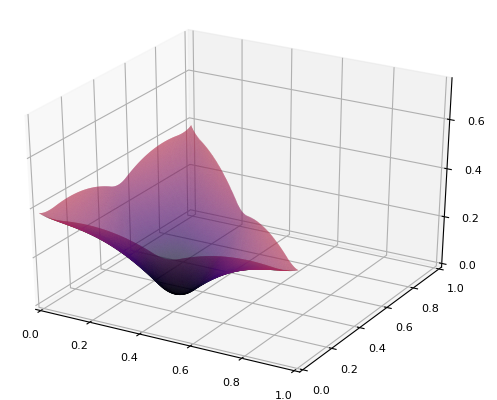}
\includegraphics[width=.492\textwidth]{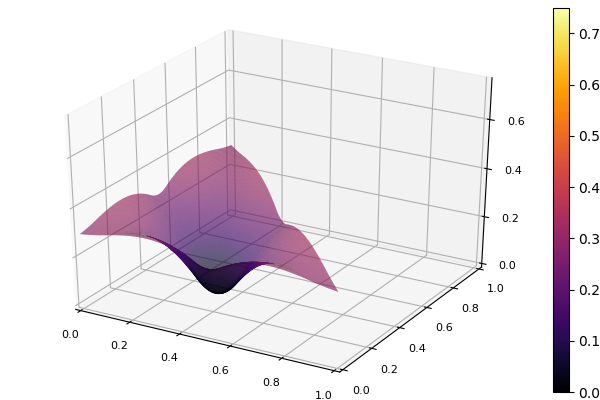}
\caption{In dimensions $d=3$, we parametrize $\sigma = \diag(x,y,1-x-y)$, and plot $(x,y) \mapsto \Delta_\eps(\sigma)$ for $\eps = 0.1$, with $H_{(h,\phi)} = S_\alpha$, the R\'enyi entropy. That is, above each $(x,y)$ in the $xy$-plane, the value of $\Delta_\eps(\diag(x,y,1-x-y))$ is plotted.  Top row, left: $\alpha = \frac{1}{2}$, right: $\alpha = 1$. Bottom row, left: $\alpha = 1.5$,  right: $\alpha = 2$. The three points $(0,0)$, $(0,1)$, $(1,0)$ in the $xy$-plane correspond to the pure states $\diag(0,0,1)$, $\diag(0,1,0)$, and $\diag(1,0,0)$, respectively. The central point $(\frac{1}{3},\frac{1}{3})$ corresponds to the completely mixed state $\tau = \frac{1}{3}\one$. We observe for $\alpha =\frac{1}{2}$ and $\alpha=1$ the maximum of $\Delta_\eps$ appears to occur at the pure states. On the other hand, for $\alpha=1.5$, the maximum is along the boundary (i.e.~for a state $\sigma$ with exactly one zero eigenvalue), and for $\alpha = 2$, the maximum occurs at states without any zero eigenvalues.} \label{fig:Delta_eps_for_Renyi}
\end{figure}

Using these results, the proof of \Cref{thm:hphi-GCB} is completed as follows. Let $\psi$ be any pure state, $\psi\in \Dpure(\cH)$. Then for any $\omega\in \cD(\cH)$, we have $\omega \prec \psi$. Therefore, by \Cref{thm:Delta_eps_Schur_convex}, we have $\Delta_\eps(\omega)\leq \Delta_\eps(\psi)$, for any $\omega\in \cD(\cH)$, and in particular for $\omega\in \{\rho,\sigma\}$. Therefore, by \eqref{eq:UBHrho-Hsig_by_max} we have
\[
|H_{(h,\phi)} (\rho) - H_{(h,\phi)}(\sigma)| \leq \Delta_\eps(\psi).
\]
 By computing $\Delta_\eps(\psi)$ using the form given in \Cref{prop:properties_of_Lambda_eps}\ref{item:Lambda-eps-on-psi}, we obtain the right-hand side of \cref{eq:hphi-uniform-bound}. 

 It remains to check under which conditions equality occurs in \eqref{eq:hphi-uniform-bound}. Assume without loss of generality that $H_{(h,\phi)}(\rho) \geq H_{(h,\phi)}(\sigma)$. Equality in \eqref{eq:UBHrho-Hsig_by_max} is equivalent to $\sigma\in \Dpure(\cH)$ by \Cref{cor:Delta-eps-max-at-pure}. Next, since the $(h,\phi)$-entropy is strictly Schur concave and $\sigma_\eps^* \prec \rho$, equality in \eqref{eq:Hrho-Hsig} is equivalent to the fact that $\rho$ is unitarily equivalent to $\sigma_\eps^*$.  The expression for $\sigma_\eps^*$ when $\sigma\in \Dpure(\cH)$ is given in \Cref{prop:properties_of_Lambda_eps}\ref{item:Lambda-eps-on-psi}. This completes the proof.\hfill\proofSymbol

\Cref{thm:Delta_eps_Schur_convex} does not extend to the $\alpha$-R\'enyi entropy for $\alpha > 1$, in which case $h$ is convex and $\phi$ strictly convex. This is discussed in the remark following \Cref{lem:reduce-to-id}, and is illustrated in \Cref{fig:Delta_eps_for_Renyi}.

\section{The majorization-minimizer map $\mmm_\eps$} \label{sec:Lambda-eps}

In order to prove \Cref{thm:hphi-GCB}, we need to use properties of the majorization-minimizer map $\mmm_\eps$ introduced in \eqref{eq:def_mmm-eps}. Let $\sigma\in \D(\cH)$ and $\eps\in (0,1]$. We formulate the definition of $\mmm_\eps$ by constructing $\sigma_\eps^*$. Note that the following is a reformulation of Lemma 4.1 of \cite{HD17}. For notational simplicity, we often suppress dependence on $\sigma$ and $\eps$ in this section, and write $\lambda_j = \lambda_j(\sigma)$ so that the eigenvalues of $\sigma$ are $\lambda_1 \geq \lambda_2 \geq \dotsm \geq \lambda_d$.

We first define a quantity $\gamma_+^{(m)}\equiv \gamma_+^{(m)}(\sigma,\eps)$, for $m\in \{0,1,\dotsc,d-1\}$, as follows
\[
\gamma_+^{(m)} := \begin{cases}
\frac{1}{m}\left(\sum_{i=1}^{m} \lambda_i - \eps\right) & \text{if }T(\sigma,\tau) > \eps \text{ and } m\neq 0\\
\frac{1}{d} & \text{else}.
\end{cases}
\]
Similarly, a quantity $\gamma_-^{(m)} \equiv \gamma_-^{(m)}(\sigma,\eps) $ is defined by
\[
\gamma_-^{(m)}:= \begin{cases}\frac{1}{m}\left(\sum_{i=d-m+1}^d \lambda_i + \eps\right) & \text{if } T(\sigma,\tau) > \eps \text{ and } m\neq 0\\
\frac{1}{d} & \text{else}.
\end{cases}
\]
Then for $\sigma \neq \tau$, we define $m_+ = m_+(\sigma,\eps)$ as the unique solution to the following inequalities:
\begin{equation}\label{eq:m_is_sol_to_this}
 \lambda_{m+1} \leq \gamma_+^{(m)} < \lambda_m, \qquad m\in \{1,\dotsc,d-1\} 
\end{equation}
and we set $m_+(\tau,\eps) = 0$.
Similarly, for $\sigma\neq \tau$, we define $m_- = m_-(\sigma,\eps)$ as the unique solution to the inequalities:
\begin{equation} \label{eq:n_is_sol_to_this}
 \lambda_{d-m+1}< \gamma_-^{(m)}\leq \lambda_{d-m}, \qquad  m\in \{1,\dotsc,d-1\}
\end{equation}
and set $m_-(\tau,\eps) = 0$. Finally, we set $\gamma_+ = \gamma_+(\sigma,\eps) := \gamma_+^{(m_+)}$ and  $\gamma_-=\gamma_-(\sigma,\eps) := \gamma_-^{(m_-)}$.

Given the eigen-decomposition $\sigma = \sum_{i=1}^d \lambda_i \ketbra{i}{i}$, we define
\begin{equation} \label{def:Lambdaeps}
\mmm_\eps(\sigma) := \sigma_\eps^* := \sum_{i=1}^{m_+} \gamma_+\ketbra{i}{i} + \sum_{i=m_++1}^{d-m_-} \lambda_i\ketbra{i}{i} + \sum_{i=d-m_-+1}^d \gamma_-\ketbra{i}{i}.
\end{equation}
To summarize, we construct $\sigma_\eps^*$ as follows: we decrease the $m_+$ largest eigenvalues of $\sigma$ by setting them to $\gamma_+$ (where $m_+$ and $\gamma_+$ are related by \cref{eq:m_is_sol_to_this}), increase the $m_-$ smallest eigenvalues of $\sigma$ by setting them to $\gamma_-$ (where $m_-$ and $\gamma_-$ are related by \cref{eq:n_is_sol_to_this}), and we keep the other eigenvalues of $\sigma$ unchanged. This is illustrated in \Cref{fig:sig-levels}, for a state $\sigma\in \cD(\cH)$ with $\eps = 0.07$ and $d=12$.
\begin{figure}[ht]
\centering
\includegraphics{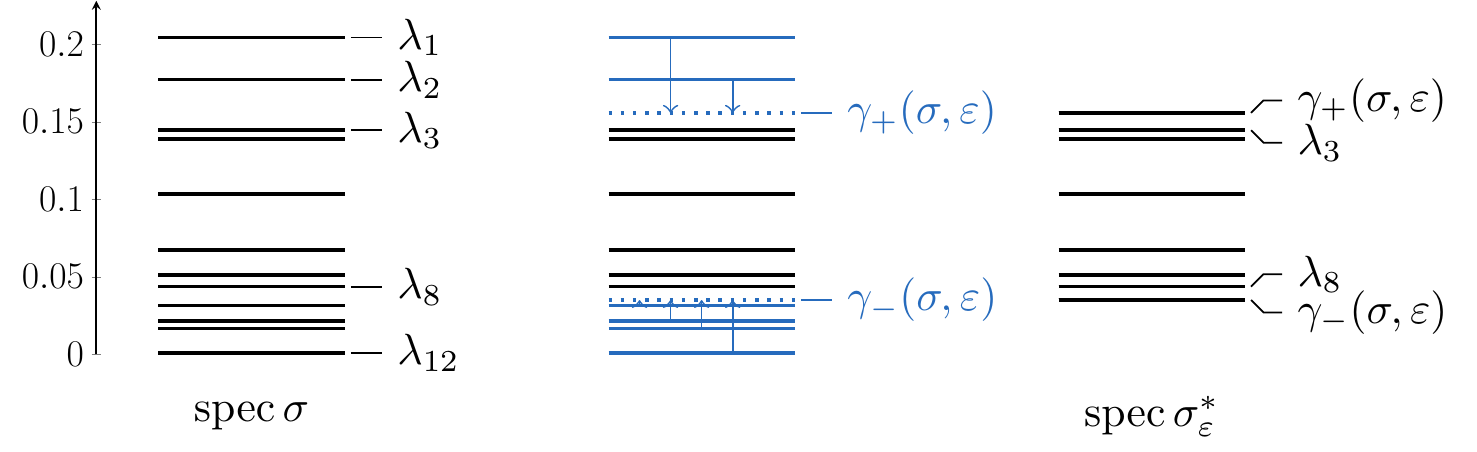}
\caption{We choose $d=12$, a state $\sigma \in \cD(\cH)$, and $\eps=0.07$, for which $m_+=2$ and $m_-=4$. Left: the eigenvalues $\lambda_1 \geq \lambda_2\geq \dotsc \geq \lambda_d$ of $\sigma$ are plotted. Center: the smallest four eigenvalues of $\sigma$ are increased to $\gamma_- = \frac{1}{4}[\lambda_1 + \dotsm + \lambda_4+  \eps]$, and the largest two eigenvalues of $\sigma$ decreased to $\gamma_+ = \frac{1}{2}[\lambda_1+\lambda_2 - \eps]$. Right: the eigenvalues of $\sigma_\eps^*$ are $\gamma_+$ with multiplicity two, $\lambda_3,\lambda_4,\dotsc,\lambda_{d-4}$, and $\gamma_-$ with multiplicity four. \label{fig:sig-levels}}
\end{figure}

Considered as a map on $\cD(\cH)$, $\mmm_\eps$ has several useful properties which are presented in the following proposition. It should be noted, however, that $\mmm_\eps$ is \emph{not} a linear map.

\begin{proposition}[Properties of $\mmm_\eps$] \label{prop:properties_of_Lambda_eps} Let $\sigma\in \D(\cH)$. We have the following properties of $\mmm_\eps$, for any $\eps\in (0,1]$.
\begin{enumerate}[label*=\alph*.,ref=(\alph*)]
	\item \label{item:Lambda_epsstates-to-states} Maps states to states: $\mmm_\eps : \D(\cH) \to \D(\cH)$. 
	\item \label{item:Lambda_eps_min_maj_order} Minimal in majorization order:   $\mmm_\eps(\sigma)\in \Be(\sigma)$ and for any $\omega\in \Be(\sigma)$, we have $\mmm_\eps(\sigma) \prec \omega$.
\item  \label{item:Lambda_eps-semigroup} Semi-group property:
if $\eps_1,\eps_2\in (0,1]$ with $\eps_1+\eps_2 \leq 1$, we have $\mmm_{\eps_1+\eps_2}(\sigma) = \mmm_{\eps_1} \circ \mmm_{\eps_2}(\sigma)$.
\item \label{item:Lambda_eps-maj-preserving} Majorization-preserving: let $\rho\in \D(\cH)$ such that $\rho \prec \sigma$. Then $\mmm_\eps(\rho) \prec \mmm_\eps(\sigma)$.
\item \label{item:Lambda_eps_fixed_point} $\tau = \frac{\one}{d}$ is the unique fixed point of $\mmm_\eps$, i.e.~the unique solution to $\sigma = \mmm_\eps(\sigma)$ for $\sigma \in \D(\cH)$. Moreover, for any $\sigma \neq \tau$, $\mmm_\eps(\sigma)$ is not unitarily equivalent to $\sigma$. 
\item \label{item:Lambda-eps-near-tau}For any state $\sigma \in \Be(\tau)$, we have $\mmm_\eps(\sigma) = \tau$.
\item \label{item:Lambda-eps-on-psi}For any pure state $\psi\in \Dpure(\cH)$, the state $\mmm_\eps(\psi)$ has the form
\begin{equation}
\mmm_\eps(\psi) = \begin{cases} 
\diag(1- \eps, \frac{\eps}{d-1},\dotsc \frac{\eps}{d-1}) & \eps < 1 - \frac{1}{d}\\
\tau := \frac{\one}{d} & \eps \geq 1 - \frac{1}{d}.
\end{cases} \label{eq:Lambda-eps-pure}
\end{equation}
\end{enumerate}
 
\end{proposition}
The proof of properties \ref{item:Lambda_epsstates-to-states} and \ref{item:Lambda_eps_min_maj_order} can be found in \cite{HD17,HO17approxmaj}; the property \ref{item:Lambda_eps-semigroup} was proved in in \cite{HD17}, property \ref{item:Lambda_eps-maj-preserving} can be found in Lemma 2 of \cite{HO17approxmaj}. The property \ref{item:Lambda_eps_fixed_point} can be shown as follows. $\mmm_\eps(\rho)$ is not unitarily equivalent to $\rho$ for $\rho\neq\tau$ follows from the construction presented above, in particular, the fact that the eigenvalues of $\mmm_\eps(\rho)$ differ from $\rho$. One immediately has that $\tau$ is a fixed point of $\mmm_\eps$, and uniqueness follows from the fact that $\mmm_\eps(\sigma)$ is not unitarily equivalent to $\sigma$ for $\sigma\neq \tau$.  Lastly, the properties \ref{item:Lambda-eps-near-tau} and \ref{item:Lambda-eps-on-psi} follow from the construction given above.

 \section{Proof of \Cref{thm:Delta_eps_Schur_convex} \label{sec:proof-Delta-eps-Schur-convex}}

\subsection{Reducing to $h=\id$}
Our first task is to reduce to the case when $h=\id$, i.e. $h(x) = x$ for all $x\in \R$. Fix $\eps \in (0,1]$ and $\rho,\sigma \in \D(\cH)$ such that $\rho\prec \sigma$ and $\rho$ and $\sigma$ are not unitarily equivalent. Let us define four variables
\begin{gather*}	
a := H_{(\id,\phi)} \circ \mmm_{ \eps} (\rho) ,\qquad b :=  H_{(\id,\phi)}(\rho), \qquad
c := H_{(\id,\phi)} \circ \mmm_{\eps} (\sigma) ,\qquad d :=  H_{(\id,\phi)}(\sigma) 
\end{gather*}
which are non-negative real numbers.
\Cref{thm:Delta_eps_Schur_convex} is the statement that
\begin{equation} \label{eq:h-abcd}
h(a) - h(b) \leq h(c) - h(d).
\end{equation}
\begin{lemma}Let $h$ be concave, and $\phi$ strictly concave.
If $a-b \leq c-d$, then \eqref{eq:h-abcd} holds. Moreover, if $h$ is strictly concave, then \eqref{eq:h-abcd} holds with strict inequality.
\label{lem:reduce-to-id}
\end{lemma}
\begin{proof}	
By the strict Schur concavity of the $(\id,\phi)$-entropy, we have  $b<a$ and $d<c$, and by \Cref{prop:properties_of_Lambda_eps} \ref{item:Lambda_eps_min_maj_order}, we have $b> d$ and $a> c$. Therefore, since $h$ is concave, we apply \Cref{cor:main_concave_ieq} to obtain
\[
\frac{h(b) - h(d)}{b-d} \geq \frac{h(a) - h(c)}{a-c}.
\]
That is,
\[
[h(b) - h(d)] \frac{a-c}{b-d} \geq h(a) -h(c)
\]
Since we have $a-c \leq b - d$ using the assumption, then $\frac{a-c}{b-d} \leq 1$, and therefore
\[
h(b) - h(d) \geq h(a) -h(c)
\]
and adding $h(c) - h(b)$ to each side yields \eqref{eq:h-abcd}.
\end{proof}
Therefore, it remains to establish $a-b\leq c-d$, which is \Cref{thm:Delta_eps_Schur_convex} when $h=\id$.

\begin{remark}An extension of \Cref{thm:hphi-GCB} to treat the $\alpha$-R\'enyi entropy for $\alpha > 1$ would need to address the case in which $h$ is convex and strictly decreasing, and $\phi$ is strictly convex. In this case, $\rho \mapsto \tr \phi(\rho)$ is Schur convex, and we have  $a<b$, $c<d$, $b<d$, and $a<c$. The analog to \Cref{lem:reduce-to-id} would be to show that $a-b \geq c-d$ implies \eqref{eq:h-abcd}. However, repeating the proof of \Cref{lem:reduce-to-id} in this case yields e.g.
\[
[h(b)- h(d)] \frac{c-a}{d-b} \leq h(a)-h(c)
\]
which is inconclusive in showing \eqref{eq:h-abcd} when $a-b \geq c-d$. This is the technical reason this proof does not extend to the $\alpha$-R\'enyi entropy for $\alpha > 1$. 

In fact, the associated quantity $\Delta_\eps$ for an $\alpha$-R\'enyi entropy with $\alpha>1$ is not Schur convex. For the example stated after \Cref{thm:Delta_eps_Schur_convex}, it can be shown that choosing $H_{(h,\phi)} = S_\alpha$ for $\alpha > 1$ yields $\Delta_\eps(\rho) > \Delta_\eps(\sigma)$.
\end{remark}

\subsection{The case $h=\id$}

We prove \Cref{thm:Delta_eps_Schur_convex} in several steps. First, we use the semigroup property of $\mmm_\eps$ to decompose $\Delta_{\eps}$ for $\eps=\eps_1+\eps_2$ in terms of $\eps_1$ and $\eps_2$ in \Cref{lem:Delta_eps1_plus_eps2}. Then we define a quantity $\delta(\rho,\sigma)$ in \Cref{def:delta-rho-sigma} such that for $\eps \leq \delta(\rho,\sigma)$, we can show that $\Delta_\eps(\rho) \leq \Delta_\eps(\sigma)$ if $\rho\prec \sigma$ (\Cref{lem:Delta_eps_Schur_convex_deltarho-sig}), using properties of $\delta(\rho,\sigma)$ presented in \Cref{lem:delta_rho}. Finally, we show  that for arbitrary $\eps\in (0,1]$, we can use \Cref{lem:Delta_eps1_plus_eps2} finitely many times to prove \Cref{thm:Delta_eps_Schur_convex}. We state the lemmas here but defer their proofs to \Cref{sec:proof_lemmas}.

\begin{lemma} \label{lem:Delta_eps1_plus_eps2}
Let $\rho\in \D(\cH)$, and $\eps_1,\eps_2\in (0,1]$ with $\eps_1+\eps_2 \leq 1$. Then
\[
\Delta_{\eps_1+\eps_2}(\rho) = \Delta_{\eps_1} \circ \mmm_{\eps_2}(\rho) + \Delta_{\eps_2}(\rho).
\]
\end{lemma}

\begin{definition}[$\delta(\rho,\sigma)$] \label{def:delta-rho-sigma}Let $\rho\in \D(\cH)$ for $\rho \neq \tau$.
Let $\mu_1 > \mu_2 > \dotsm > \mu_\ell$ denote the distinct ordered eigenvalues of $\rho$, and define 
\begin{equation}
\delta(\rho) =\min \{k_+(\rho) (\mu_1-\mu_2), k_-(\rho)(\mu_{\ell-1}-\mu_\ell) \}. \label{eq:delta_rho}
\end{equation}
For $\rho,\sigma\in\D(\cH)$ with $\rho\neq\tau\neq\sigma$, define
\begin{equation}
\delta(\rho,\sigma) = \min \{ \delta(\rho), \delta(\sigma) \}.
\label{eq:def-delta-rho-sigma}
\end{equation}
\end{definition}

For any $\eps \leq \delta(\rho,\sigma)$, the map $\mmm_\eps$ only ``moves'' the largest and smallest eigenvalue of $\rho$ and of $\sigma$, as shown by the following result and illustrated through an example in \Cref{fig:Lambda-eps}.
\begin{figure}[ht]
\centering
\includegraphics[width=.75\textwidth]{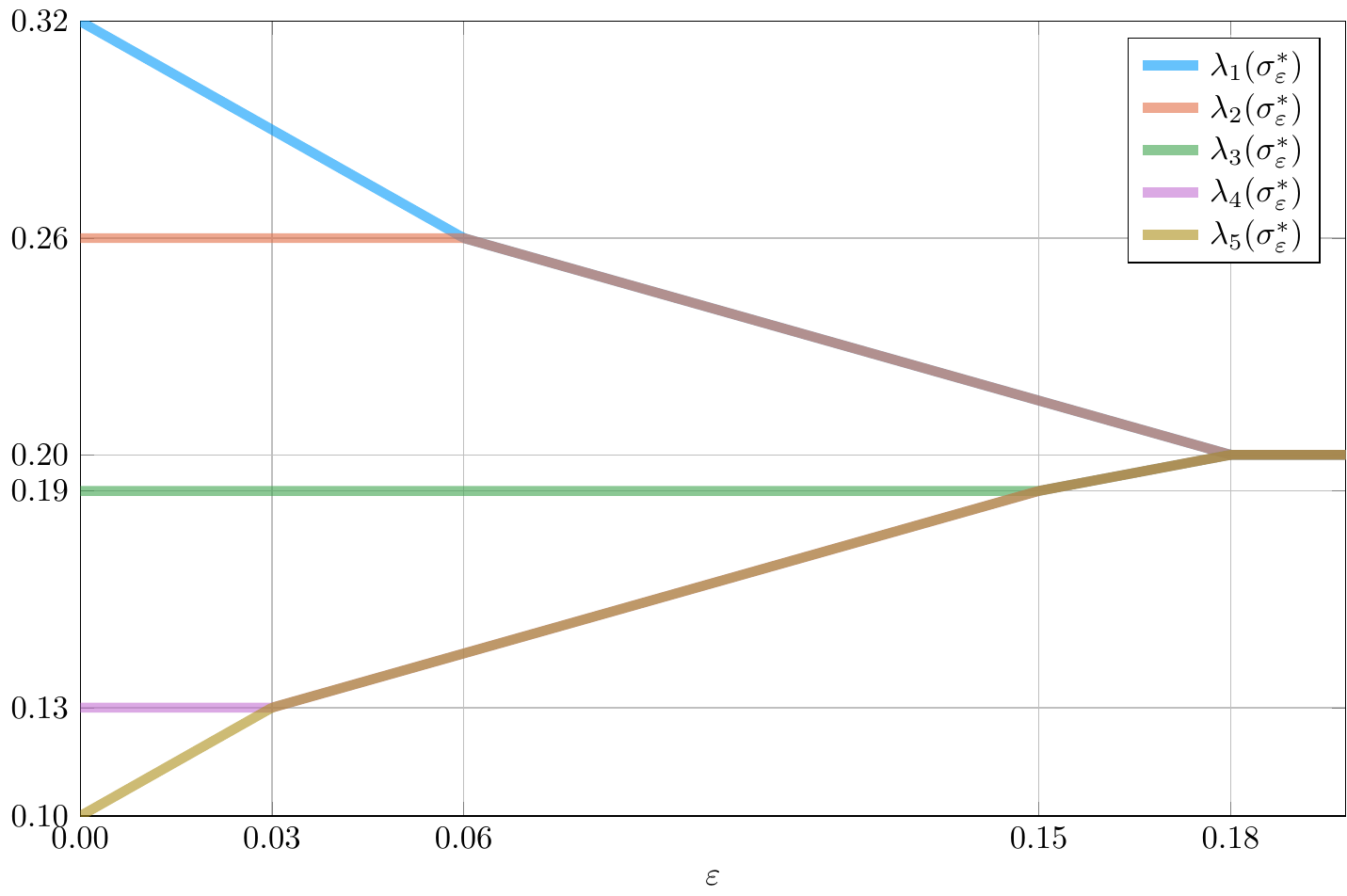}
\caption{For the 5-dimensional state $\sigma = \diag(0.32, 0.26, 0.19, 0.13, 0.10)$, the spectrum of $\sigma_\eps^* = \mmm_\eps(\sigma)$ is plotted as a function of $\eps$. This plot is a continuous (in $\eps$) analog to the type of plot shown in \Cref{fig:sig-levels}, which shows the spectrum of $\sigma_\eps^*$ at two discrete points, $\eps = 0$ and $\eps= 0.07$, in a different example. 
Here, at $\eps=0$, the five lines correspond to the five eigenvalues of $\sigma$, each with multiplicity one. For $\eps \leq 0.03$, $\sigma_\eps^* = \diag(0.32 -\eps, 0.26, 0.19, 0.13, 0.10+\eps)$ and differs from $\sigma$ only in the smallest and largest eigenvalue. When $\eps$ reaches $0.03$, the multiplicity of the smallest eigenvalue of $\sigma_\eps^*$ increases to 2. Between $\eps=  0.03$ and $\eps = 0.06$, again only the smallest and largest eigenvalues change, but the smallest eigenvalue has multiplicity 2. This process continues until every eigenvalue reaches $\frac{1}{d}=0.2$ at $T(\sigma,\tau) = 0.18$.} \label{fig:Lambda-eps}
\end{figure}

\begin{lemma} \label{lem:delta_rho}
Let $\rho \neq \tau$. For any $\eps \leq \delta(\rho)$, we have
\[
m_+(\rho,\eps) = k_+(\rho), \quad\text{and}\quad m_-(\rho,\eps) = k_-(\rho).
\]
Moreover, if $\eps = \delta(\rho)$ then either $k_+ (\mmm_{\eps}(\rho)) > k_+ (\rho)$ or $k_- (\mmm_{\eps}(\rho)) > k_- (\rho)$.
\end{lemma}
Using this result, we can prove the Schur convexity of $\Delta_\eps$ for $\eps$ small enough (depending on $\rho$ and $\sigma$).
\begin{lemma}\label{lem:Delta_eps_Schur_convex_deltarho-sig}
Let $\rho,\sigma \in \D(\cH)$ with $\rho \prec \sigma$. Let $\eps \leq \delta(\rho,\sigma)$, defined by \eqref{eq:def-delta-rho-sigma}. Then
\begin{equation} \label{eq:Delta_schur_convex_eps_small}
\Delta_\eps(\rho) \leq \Delta_\eps(\sigma).
\end{equation}
Moreover, equality in \eqref{eq:Delta_schur_convex_eps_small} implies that $\lambda_{\pm}(\rho) = \lambda_\pm(\sigma)$.
\end{lemma}

We can iterate this result using \Cref{lem:Delta_eps1_plus_eps2} to prove \Cref{thm:Delta_eps_Schur_convex} in general.
\paragraph{Proof of \Cref{thm:Delta_eps_Schur_convex}}
Let $\rho,\sigma\in \D(\cH)$ and $\eps\in(0,1]$.  Note that if $\sigma = \tau$, then $\rho \prec \sigma$ implies $\rho = \tau$, and hence $\Delta_\eps(\rho) = 0 = \Delta_\eps(\sigma)$. If $\rho = \tau \neq \sigma$, then $\Delta_\eps(\rho) = 0 < \Delta_\eps(\sigma)$ by the strict Schur concavity of the $H_{(h,\phi)}$ entropy. Therefore, we can assume $\rho \neq \tau\neq \sigma$.
\begin{enumerate}[{Step }1.]
	\item  Set $\rho_1=\rho$ and $\sigma_1=\sigma$. If $\eps \leq \delta_1:=\delta(\rho_1,\sigma_1)$, we conclude via \Cref{lem:Delta_eps_Schur_convex_deltarho-sig}. Otherwise, set $\eps_1 = \eps- \delta_1$. Then by \Cref{lem:Delta_eps1_plus_eps2},
\[
\Delta_{\eps_1+\delta_1}(\rho_1) = \Delta_{\eps_1} \circ \mmm_{\delta_1}(\rho_1) + \Delta_{\delta_1}(\rho_1)
\]
and
\[
\Delta_{\eps_1+\delta_1}(\sigma_1) = \Delta_{\eps_1} \circ \mmm_{\delta_1}(\sigma_1) + \Delta_{\delta_1}(\sigma_1).
\]
 We invoke \Cref{lem:Delta_eps_Schur_convex_deltarho-sig} to find $\Delta_{\delta_1}(\rho_1) \leq \Delta_{\delta_1}(\sigma_1)$; it remains to show
\[
\Delta_{\eps_1} ( \mmm_{\delta_1}(\rho_1))\leq \Delta_{\eps_1} ( \mmm_{\delta_1}(\sigma_1)).
\]
\item Set $\rho_2 = \mmm_{\delta_1}(\rho_1)$ and $\sigma_2 = \mmm_{\delta_1}(\sigma_1)$. If either $\rho_2 = \tau$ or $\sigma_2 = \tau$ we conclude by the argument presented at the start of the proof. Otherwise, we set $\delta_2 := \delta(\rho_2,\sigma_2)$ and proceed.

If $\eps_1 \leq \delta_2$, we conclude by \Cref{lem:Delta_eps_Schur_convex_deltarho-sig}. Otherwise, set $\eps_2 = \eps_1 - \delta_2$, expand e.g. 
\[
\Delta_{\eps_2+\delta_2}(\rho_2) = \Delta_{\eps_2} \circ \mmm_{\delta_2}(\rho_2) + \Delta_{\delta_2}(\rho_2), \qquad 
\Delta_{\eps_2+\delta_2}(\sigma_2) = \Delta_{\eps_2} \circ \mmm_{\delta_2}(\sigma_2) + \Delta_{\delta_2}(\sigma_2), \qquad 
\]
and conclude $\Delta_{\delta_2}(\rho_2) \leq \Delta_{\delta_2}(\sigma_2)$ by \Cref{lem:Delta_eps_Schur_convex_deltarho-sig}. It remains to show
\[
\Delta_{\eps_2} \circ \mmm_{\delta_2}(\rho_2) \leq \Delta_{\eps_2} \circ \mmm_{\delta_2}(\sigma_2).
\]

\item[Step $k$.] We continue recursively: for $k\geq 3$,  we define $\rho_k =\mmm_{\delta_{k-1}}(\rho_{k-1})$, $\sigma_k = \mmm_{\delta_{k-1}}(\sigma_{k-1})$. If either $\rho_k = \tau$ or $\sigma_k = \tau$, we conclude as before; otherwise, set $\delta_k = \delta(\rho_k,\sigma_k)$. If $\eps_{k-1} \leq \delta_k$, we conclude by \Cref{lem:Delta_eps_Schur_convex_deltarho-sig}; otherwise, define $\eps_k = \eps_{k-1} - \delta_k$, expand by \Cref{lem:Delta_eps1_plus_eps2} to find
\[
\Delta_{\eps_k+\delta_k}(\rho_k) = \Delta_{\eps_k} \circ \mmm_{\delta_k}(\rho_k) + \Delta_{\delta_k}(\rho_k), \qquad 
\Delta_{\eps_k+\delta_k}(\sigma_k) = \Delta_{\eps_k} \circ \mmm_{\delta_k}(\sigma_k) + \Delta_{\delta_k}(\sigma_k),
\]
and conclude
$\Delta_{\delta_k}(\rho_k) \leq \Delta_{\delta_k}(\sigma_k)$ by \Cref{lem:Delta_eps_Schur_convex_deltarho-sig}. At the end of step $k$, it remains to show that
\[
\Delta_{\eps_k} \circ \mmm_{\delta_k}(\rho_k) \leq \Delta_{\eps_k} \circ \mmm_{\delta_k}(\sigma_k).
\]
\end{enumerate}

This process must terminate in less than $4d$ steps, as follows. At each step $k$ for which the process does not conclude, we have either $\delta(\rho_k,\sigma_k) = \delta(\rho_k)$ and therefore  $k_+(\rho_k) > k_+(\rho_{k-1})$ or $k_-(\rho_k) > k_-(\rho_{k-1})$ or else $\delta(\rho_k,\sigma_k) = \delta(\sigma_k)$ and therefore $k_+(\sigma_k) > k_+(\sigma_{k-1})$ or $k_-(\sigma_k) > k_-(\sigma_{k-1})$, by \Cref{lem:delta_rho}. Since $k_\pm(\omega)\leq d$ for $\omega\in \D(\cH)$ and one of each of the four integers $k_\pm(\rho), k_\pm(\sigma)$ increases at each step,  there cannot be more than $4d$ steps in total. Note that $\Delta_\eps(\rho) = \Delta_\eps(\sigma)$ implies equality in the use of \Cref{lem:Delta_eps_Schur_convex_deltarho-sig} in Step 1, which requires $\lambda_+(\rho) = \lambda_+(\sigma)$.\hfill\proofSymbol 

\section{Proof of lemmas \label{sec:proof_lemmas}}
\paragraph{Proof of \Cref{lem:Delta_eps1_plus_eps2}}
We expand
\[
\Delta_{\eps_1+\eps_2}(\rho) = H_{(h,\phi)} \circ \mmm_{\eps_1+\eps_2}(\rho) - H_{(h,\phi)}(\rho).
\]
By \Cref{prop:properties_of_Lambda_eps}\ref{item:Lambda_eps-semigroup}, we have $\mmm_{\eps_1+\eps_2}(\rho) = \mmm_{\eps_1}\circ \mmm_{\eps_2}(\rho)$. Thus,
\begin{align*}	
\Delta_{\eps_1+\eps_2}(\rho) &= H_{(h,\phi)} \circ\mmm_{\eps_1}( \mmm_{\eps_2}(\rho))  - H_{(h,\phi)}( \mmm_{\eps_2}(\rho)) + H_{(h,\phi)}( \mmm_{\eps_2}(\rho))  - H_{(h,\phi)}(\rho)\\
&= \Delta_{\eps_1} ( \mmm_{\eps_2}(\rho)) + \Delta_{\eps_2}(\rho). \tag*{\proofSymbol}
\end{align*}

\paragraph{Proof of \Cref{lem:delta_rho}}
We use the notation of \Cref{def:delta-rho-sigma}.	 We check that the choice $m=k_+(\rho)$ satisfies the definition of $m_+(\rho,\eps)$, namely that the choice $m=k_+(\rho)$ solves \eqref{eq:m_is_sol_to_this}. 

\begin{itemize}
	\item If $T(\rho,\tau) > \eps$, then $\gamma_+^{(m)}(\rho,\eps) = \frac{1}{m}\left(\sum_{i=1}^{m} \lambda_i(\rho) - \eps\right)$.
And indeed, taking $m = k_+(\rho)$ we find
\[
\lambda_{k_++1}(\rho) = \mu_2 \leq \frac{1}{k_+}\left(\sum_{i=1}^{k_+} \lambda_i(\rho) - \eps\right) = \mu_1 - \frac{\eps}{k_+}
\]
since $\frac{\eps}{k_+} \leq \frac{1}{k_+}\delta(\rho,\sigma) \leq \mu_1-\mu_2$. Additionally, $ \mu_1 - \frac{\eps}{k_+} < \mu_1 = \lambda_{k_+}(\rho)$. Therefore, $m = k_+(\rho)$ solves \eqref{eq:m_is_sol_to_this}, hence $m_+(\rho,\eps) = k_+(\rho)$.

\item In the case $0 < T(\rho,\tau) \leq \eps$. Then $\gamma_+(\rho,\eps) = \frac{1}{d}$. Since $\rho \neq \tau$, we have $\lambda_{k_+(\rho)}(\rho) = \mu_1 > \frac{1}{d}$. Moreover,
\[
k_+(\rho) \big(\mu_1 - \tfrac{1}{d}\big) \leq \tr [(\rho - \tau)_+ ] = T(\rho,\tau) \leq \eps \leq k_+(\rho) (\mu_1 - \mu_2)
\]
and therefore $\mu_1 - \frac{1}{d}\leq \mu_1 - \mu_2$, yielding $\mu_2 \leq \frac{1}{d}$. Thus, $m_+(\rho,\eps) = k_+(\rho)$.
\end{itemize}
Proving that $m_-(\rho,\eps) = k_-(\rho)$ is analogous.

Next, consider $\eps = \delta(\rho)$. If $0 < T(\rho,\tau) \leq \eps$, then $\mmm_\eps(\rho) = \tau$ (by \Cref{prop:properties_of_Lambda_eps}\ref{item:Lambda-eps-near-tau}) and $d = k_+(\mmm_\eps(\rho)) > k_+(\rho)$, by the assumption that $\rho\neq \tau$. 
Otherwise, without loss of generality, assume $\delta(\rho,\sigma) = k_+(\rho) (\mu_1-\mu_2)$. We show that $k_+(\mmm_\eps(\rho)) > k_+(\rho)$. By the above, $m_+(\rho,\eps) = k_+(\rho)$, and therefore
\[
\gamma_+(\rho,\eps)= \mu_1 + \frac{\eps}{k_+(\rho)} =  \mu_1 +(\mu_1-\mu_2) = \mu_2.
\]
As
\[
\mmm_\eps(\rho) = 
\displaystyle \sum_{i=1}^{m_+(\rho,\eps)} \gamma_+(\rho,\eps)\ketbra{i}{i} + \sum_{i=m+1}^{d-m_-(\rho,\eps)} \lambda_i(\rho)\ketbra{i}{i} + \sum_{i=d-m_-(\rho,\eps)+1}^d \gamma_-(\rho,\eps)\ketbra{i}{i} 
\]
by \cref{def:Lambdaeps} and $\lambda_{m_+(\rho,\eps)+1}(\rho) = \mu_2$, we have that $k_+(\mmm_\eps(\rho))$, the multiplicity of $\mu_2$ for $\mmm_\eps(\rho)$, is strictly larger than $k_+(\rho) =m_+(\rho,\eps)$. \hfill\proofSymbol

\paragraph{Proof of \Cref{lem:Delta_eps_Schur_convex_deltarho-sig}}
As in the proof of \Cref{thm:Delta_eps_Schur_convex},  if $\sigma = \tau$, then $\rho \prec \sigma$ implies $\rho = \tau$, and hence $\Delta_\eps(\rho) = 0 = \Delta_\eps(\sigma)$. If $\rho = \tau \neq \sigma$, then $\Delta_\eps(\rho) = 0 < \Delta_\eps(\sigma)$ by the strict Schur concavity of the $H_{\id,\phi}$ entropy.
Now, assume $\rho \neq \tau \neq \sigma$. We aim to show
\begin{equation}
H_{(\id,\phi)} \circ \mmm_\eps(\rho) - H_{(\id,\phi)}  (\rho) \leq H_{(\id,\phi)} \circ \mmm_\eps(\sigma) - H_{(\id,\phi)}  (\sigma) .\label{eq:WTS_diffHs}
\end{equation}
By two applications of \Cref{lem:delta_rho}, we have $m_+(\rho,\eps) = k_+(\rho)$, $m_-(\rho,\eps) = k_-(\rho)$, $m_+(\sigma,\eps) = k_+(\sigma)$, and $m_-(\sigma,\eps) = k_-(\sigma)$. Therefore, by \eqref{eq:def_h-phi-entropy} and \eqref{def:Lambdaeps},
\[
H_{(\id,\phi)} (\mmm_\eps(\rho)) =  k_+(\rho) \phi( \gamma_+(\rho)) + \sum_{i = k_+(\rho)+1}^{d- k_-(\rho)} \phi(\lambda_i(\rho)) + k_-(\rho) \phi(\gamma_-(\rho))
\]
since $h=\id$. The $\phi(\lambda_i(\rho))$ terms for $i=k_+(\rho)+1,\dotsc,d- k_-(\rho)$ therefore cancel in $\Delta_\eps(\rho)$ yielding
\begin{equation} \label{eq:Hidphirho}
H_{(\id,\phi)} \circ \mmm_\eps(\rho) - H_{(\id,\phi)}  (\rho) = k_+(\rho) [\phi( \gamma_+(\rho,\eps)) -\phi(\lambda_+(\rho))]+ k_-(\rho)[ \phi(\gamma_-(\rho,\eps)) -\phi(\lambda_-(\rho))]
\end{equation}
and similarly
\begin{equation}\label{eq:Hidphisigma}
H_{(\id,\phi)} \circ \mmm_\eps(\sigma) - H_{(\id,\phi)}  (\sigma) = k_+(\sigma) [\phi( \gamma_+(\sigma,\eps)) - \phi(\lambda_+(\sigma))] + k_-(\sigma) [\phi(\gamma_-(\sigma,\eps)) - \phi(\lambda_-(\sigma))].
\end{equation}
We conclude by invoking \Cref{lem:ineq} below, which bounds the first term (resp.~second term) of \eqref{eq:Hidphirho} by the first term (resp.~second term) of \eqref{eq:Hidphisigma}.\hfill\proofSymbol

 \begin{lemma} \label{lem:ineq} For $\rho\prec \sigma$ with $\rho\neq\tau\neq\sigma$ and $0<\eps \leq \delta(\rho,\sigma)$, we have
 \begin{gather} \label{eq:alpha2ieq-proof}
	k_\pm(\rho)[ \phi(\gamma_\pm(\rho,\eps)) -\phi(\lambda_\pm(\rho)) ] \leq k_\pm(\sigma)[  \phi(\gamma_\pm(\sigma,\eps))-\phi(\lambda_\pm(\sigma))]
	\end{gather}
	and that equality in \eqref{eq:alpha2ieq-proof} implies  $\lambda_\pm(\rho) = \lambda_\pm(\sigma)$.
 \end{lemma}
 To prove this result, we first recall a simple consequence of the majorization order $\rho \prec \sigma$.
\begin{lemma}
If $\rho \prec \sigma$, then $T(\rho,\tau) \leq T(\sigma,\tau)$.
\end{lemma}
\begin{proof}	
If $\rho \prec \sigma$, then by Theorem 2-2 (b) of \cite{alberti-uhlmann1982}, we have $\rho = \Phi(\sigma)$ for a map $\Phi(\cdot) = \sum_i p_i U_i \cdot U_i^*$ where $p_i$ is a finite probability distribution and each $U_i$ is unitary. $\Phi$ is completely positive and trace-preserving (CPTP) as well as unital. Since $\Phi(\tau)=\tau$,
\[
T(\rho,\tau) = T(\rho,\Phi(\tau)) = T(\Phi(\sigma),\Phi(\tau)) \leq  T(\sigma,\tau)
\]
where the inequality follows from the monotonicity of the trace distance under CPTP maps.
\end{proof}

\paragraph{Proof of \Cref{lem:ineq}}
We prove the case $+$ in \cref{eq:alpha2ieq-proof}; the case $-$ is proved analogously. First, we have $\gamma_+(\rho,\eps)\leq \gamma_+(\sigma,\eps)$ and $\lambda_+(\rho) \leq \lambda_+(\sigma)$, using that $\rho \prec \sigma$ and $\mmm_\eps \rho \prec \mmm_\eps \sigma$ by \Cref{prop:properties_of_Lambda_eps}. Moreover, by definition, $\gamma_+(\rho,\eps) < \lambda_+(\rho)$ and $\gamma_+(\sigma,\eps) < \lambda_+(\sigma)$. Therefore, by applying \Cref{cor:main_concave_ieq} and multiplying by minus one, we have
\begin{equation}\label{eq:phi-ieq-proof-minus-1}
\frac{ \phi(\gamma_+(\rho,\eps)) -\phi(\lambda_+(\rho)) }{\lambda_+(\rho)-\gamma_+(\rho,\eps)} \leq  \frac{  \phi(\gamma_+(\sigma,\eps)) - \phi(\lambda_+(\sigma))}{\lambda_+(\sigma) - \gamma_+(\sigma,\eps)}.
\end{equation}
and that equality requires  $\lambda_+(\rho) = \lambda_+(\sigma)$.

Now, we complete the proof of \eqref{eq:alpha2ieq-proof} in three cases.
\begin{enumerate}
	\item[Case 1:] $T(\sigma,\tau) \leq \eps$. Then $T(\rho,\tau)\leq \eps$ as well, and
	\[
	\lambda_+(\rho) - \gamma_+(\rho,\eps) = \lambda_+(\rho) - \frac{1}{d}.
	\]
	As shown in the proof of \Cref{lem:delta_rho}, the second-largest eigenvalue of $\rho$ is less or equal to $\frac{1}{d}$. Therefore,
	\[
	T(\rho,\tau) = \tr[ (\rho -\tau)_+] = k_+(\rho)\big(\lambda_+(\rho) - \tfrac{1}{d}\big),
	\]
	and hence,
	\begin{equation}\label{eq:rho-diff-is-T_proof}
	\lambda_+(\rho) - \gamma_+(\rho,\eps) = \frac{1}{k_+(\rho)}T(\rho,\tau).
	\end{equation}

	As $T(\sigma,\tau) \leq \eps$, we likewise have $\lambda_+(\sigma) - \gamma_+(\sigma,\eps) = \frac{1}{k_+(\sigma)}T(\sigma,\tau)$.
	Then \eqref{eq:phi-ieq-proof-minus-1} yields
	\[
\frac{k_+(\rho)[ \phi(\gamma_+(\rho,\eps)) -\phi(\lambda_+(\rho)) ]}{T(\rho,\tau)} \leq  \frac{ k_+(\sigma)[  \phi(\gamma_+(\sigma,\eps))-\phi(\lambda_+(\sigma))]}{T(\sigma,\tau)}.
	\]
	Since $ T(\sigma,\tau)\geq T(\rho,\tau) $, we may bound the right-hand side by $\frac{ k_+(\sigma)[  \phi(\gamma_+(\sigma,\eps))-\phi(\lambda_+(\sigma))]}{T(\rho,\tau)}$. Then multiplying by $T(\rho,\tau)$ yields \eqref{eq:alpha2ieq-proof}.

\item[Case 2:] $T(\rho,\tau) \leq \eps < T(\sigma,\tau)$. In this case, \eqref{eq:rho-diff-is-T_proof} holds, and $\gamma_+(\sigma,\eps) = \lambda_+(\sigma) - \frac{\eps}{k_+(\sigma)}$.
 Therefore, \eqref{eq:phi-ieq-proof-minus-1} yields
\[
\frac{k_+(\rho) [\phi(\gamma_+(\rho,\eps)) -\phi(\lambda_+(\rho)) ]}{T(\rho,\tau)} \leq  \frac{  k_+(\sigma)[\phi(\gamma_+(\sigma,\eps)) - \phi(\lambda_+(\sigma))]}{\eps}.
\]
Similarly to the previous case, the inequality $\eps\geq T(\rho,\tau) $ bounds the right-hand side by $\frac{  k_+(\sigma)[\phi(\gamma_+(\sigma,\eps)) - \phi(\lambda_+(\sigma))]}{T(\rho,\tau)}$, and multiplying by $T(\rho,\tau)$ yields \eqref{eq:alpha2ieq-proof}.

\item[Case 3:] $T(\rho,\tau) > \eps$. Then $\gamma_+(\rho,\eps) = \lambda_+(\rho) - \frac{\eps}{k_+(\rho)}$, and $\gamma_+(\sigma,\eps) = \lambda_+(\sigma) - \frac{\eps}{k_+(\sigma)}$. Therefore, \eqref{eq:phi-ieq-proof-minus-1} yields
\[
\frac{k_+(\rho)[ \phi(\gamma_+(\rho,\eps)) -\phi(\lambda_+(\rho)) ]}{\eps} \leq  \frac{ k_+(\sigma)[  \phi(\gamma_+(\sigma,\eps))-\phi(\lambda_+(\sigma))]}{\eps} 
\]
and multiplying by $\eps$ yields \eqref{eq:alpha2ieq-proof}.
\end{enumerate}
Note that in each case, equality in \eqref{eq:alpha2ieq-proof} requires equality in \eqref{eq:phi-ieq-proof-minus-1}.
\hfill\proofSymbol

\appendix

\section{An elementary property of concave functions} \label{sec:elem-prop-concave-fun}

Given a function $\phi: I\to \R$ defined on an interval $I\subset \R$, we define the ``slope function,''
\[
\cs(x_1,x_2) =  \frac{\phi(x_2)- \phi(x_1)}{x_2-x_1}
\]
for $x_1,x_2\in I$ with $x_1\neq x_2$.
Note that $\cs$ is symmetric in its arguments. It can be shown that $\phi$ is concave (resp.~strictly concave) if and only if $\cs$ is monotone decreasing (resp.~strictly decreasing) in each argument.
\begin{proposition} \label{cor:main_concave_ieq}
Let $I\subset \R$ be an interval and $\phi: I \to \R$ be concave. For any $x_1,x_2,y_1,y_2\in I$ such that $x_1\neq x_2$, $y_1\neq y_2$, $x_1 \leq y_1$ and $x_2 \leq y_2$
we have
\[
\cs(x_1,x_2)\geq \cs(y_1,y_2).
\]
If $\phi$ is strictly concave, then equality is achieved if and only if $x_1 = y_1$ and $x_2 = y_2$.
\end{proposition}
\begin{proof}	
For $\phi$ concave, we have $\cs(x_1,x_2) \geq \cs(y_1,x_2) \geq \cs(y_1,y_2)$. Next, assume $\phi$ is strictly concave. Then equality holds in the first inequality if and only if $x_2 = y_2$, and in the second if and only if $x_1=y_1$, completing the proof.
\end{proof}

\printbibliography
\end{document}